\documentclass[onecolumn,twoside]{IEEEtran}
\IEEEoverridecommandlockouts

\usepackage{ifpdf}
\usepackage{cite}
\usepackage[cmex10]{amsmath}
\usepackage{enumerate}
\usepackage{amsthm}
\usepackage{latexsym}
\usepackage{amssymb}
\usepackage{xspace}
\usepackage{amscd}
\usepackage{enumerate}
\usepackage{graphicx}
\usepackage{epic}
 \numberwithin{equation}{section}
 
 \newtheorem{theorem}{Theorem}[section]
\newtheorem{lemma}[theorem]{Lemma}
\newtheorem{proposition}[theorem]{Proposition}
\newtheorem{corollary}[theorem]{Corollary}
\theoremstyle{definition}
\newtheorem{definition}[theorem]{Definition}

 \theoremstyle{remark}

\newtheorem{example}[theorem]{Example}

\newcommand{\N}{{\mathbb N}}

\newcommand{\Z}{{\mathbb Z}}
\newcommand{\F}{\mathbb F}
\renewcommand{\b}[1]{\mathbf{#1}}
\newcommand{\Le}{\mathbb L}

\newcommand{\Q}{{\mathcal Q}}
\newcommand{\D}{{\mathcal D}}
\newcommand{\lcm}{\mathrm{lcm}}
\newcommand{\tq}{\; : \;}
\newcommand{\doble}[2]{\genfrac{}{}{0cm}{2}{#1}{#2}}

\begin{document}

\title{Constructions of Abelian Codes multiplying dimension of cyclic codes
\thanks{This work was partially supported by MINECO, project MTM2016-77445-P, and Fundaci\'{o}n S\'{e}neca of Murcia, project 19880/GERM/15. The second author has been supported by Departamento Administrativo de Ciencia, Tecnolog\'{\i}a e Innovaci\'on de la Rep\'ublica de Colombia.}}

\author{\IEEEauthorblockN{Jos\'e Joaqu\'{i}n Bernal\IEEEauthorrefmark{1},
Diana H. Bueno-Carre\~no\IEEEauthorrefmark{2} and
Juan Jacobo Sim\'on\IEEEauthorrefmark{1}. \\
\IEEEauthorblockA{\IEEEauthorrefmark{1}Departamento de Matem\'{a}ticas\\
Universidad de Murcia,
30100 Murcia, Spain.\\ Email: \{josejoaquin.bernal, jsimon\}@um.es} \\
\IEEEauthorblockA{\IEEEauthorrefmark{2}Departamento de Ciencias Naturales y Matem\'{a}ticas\\
Pontificia Universidad Javeriana, 
 Cali, Colombia\\
 Email: dhbueno@javerianacali.edu.co}
}}

\maketitle

\begin{abstract}
In this note, we apply some techniques developed in \cite{BBCS2,BBCS true,BGS} to give a particular construction of bivariate Abelian Codes from cyclic codes,  multiplying their dimension and preserving their apparent distance. We show that, in the case of cyclic codes whose maximum BCH bound equals its minimum distance the obtained abelian code verifies the same property; that is, the strong apparent distance and the minimum distance coincide.  We finally use this construction to multiply Reed-Solomon codes to abelian codes.
\end{abstract}

%%% ----------------------------------------------------------------------

%%% ----------------------------------------------------------------------
%\tableofcontents

\section{Introduction}
In \cite{BBCS2}, we improve the notion and computation of the apparent distance for abelian codes given in \cite{Camion} and \cite{Evans} by means of the $q$-orbit structure of defining sets of abelian codes. These results allows us to design, based on a suitable election of $q$-orbits, abelian codes having nice bounds and parameters. In this note, we apply those techniques to construct bivariate BCH codes from cyclic codes, in such a way that we preserve apparent distance but multiplying their dimension. We show that, in the case of cyclic codes whose maximum BCH bound equals its minimum distance the obtained abelian code verifies the same property; that is, the strong apparent distance and the minimum distance coincide; in particular, this drives us to multiply dimension Reed-Solomon codes to abelian codes preserving the true minimum distance.  \textbf{As it happens with others families of abelian codes, there are alternative constructions to get this one (see, for example \cite{Jensen}). We know that each alternative construction shows different structural properties that allows us to see easily, some specific qualities or parameters; as it happens in our case with the apparent distance and the true minimum distance.}

\section{Notation and preliminaries}\label{preliminares}

In this section, we introduce the basic concepts and preliminary results. We shall restrict all notions to the bivariate abelian codes; so throughout this paper, Abelian Code will be an ideal in group algebras $\F_q G$, where $\F_q$ denotes the field with $q$ elements with $q$ a power of a prime $p$ and $G$ is an abelian group with a decomposition $G\simeq C_{r_1}\times C_{r_2}$, where $C_{r_i}$ the cyclic group of order $r_i$, for $i=1,2$. It is well-known that this decomposition induces a canonical isomorphism of $\F_q$-algebras from $\F_q G$ to 
   $$\F_q[X,Y]/\left\langle X^{r_1}-1,Y^{r_2}-1\right\rangle.$$ 
   
 We denote this quotient algebra by $\F_q(r_1,r_2)$. So, we identify the codewords with polynomials $f=f(X,Y)$ such that every monomial satisfy that the degree of the indeterminate $X$ belongs to $\Z_{r_1}$ and the degree of $Y$ belongs to $\Z_{r_2}$, where $\Z_{r_i}$ is the ring of integers modulo $r_i$, for $i=1,2$, that we always write as canonical representatives; so that, for any $a\in \Z$ we denote the canonical representative by $\overline{a}$ if it is possible that $a\not\in\{0,\dots,r_i-1\}$, otherwise we only write $\overline{a}=a$.  We deal with abelian codes in the semisimple case; that is, we always assume that $\gcd(r_i,q)=1$ for  $i=1,2$.
 
 We denote $I=\Z_{r_1}\times \Z_{r_2}$ and for $i=1,2$, we denote by $U_{r_i}$ the set of all $r_i$-th primitive roots of unity and define $U=U_{r_1}\times U_{r_2}$. It is a known fact that, for a fixed $\hat{\alpha}=(\alpha,\beta) \in U$, any abelian code $C$ is  determined by its defining set, with respect to $\hat{\alpha}$, which is defined as 
\[\D_{\hat{\alpha}}\left(C\right)= \left\{ (a,b)\in I \tq c(\alpha^{a},\beta^b)=0, \;\forall c\in C\right\}.\]

We note that, for any element  $f\in \F_q(r_1,r_2)$, viewed as a polynomial, we may also define its defining set as $\D_{\hat \alpha}(f)=\{(a,b)\in I \tq f(\alpha^a,\beta^b)=0\}$.

Given an element $\b{a}=(a,b)\in I$, we define its $q$-orbit modulo  $\left(r_1,r_2 \right)$ as the set 
$$ Q(\b a)=\left\{\left(\overline{a\cdot q^{i}} ,\overline{ b\cdot q^{i}}  \right) \in I \tq i\in \N\right\}.$$
In our case, it is known that the defining set  $\D_{\hat{\alpha}}\left(C\right)$ is a disjoint union of $q$-orbits modulo $(r_1,r_2)$. Conversely, every union of $q$-orbits modulo $(r_1,r_2)$ determines a bivariate abelian code (see \cite{BBCS2,Camion} or \cite{Imai} for details).  We recall that the notion of defining set in the case of cyclic codes corresponds to the notion of $q$-cyclotomic coset of a positive integer $a$ modulo $n$.

Let $\Le|\F_q$ be an extension field containing $U_{r_i}$, for   $i=1,2$. The discrete Fourier transform of a polynomial $f\in \F_q(r_1,r_2)$ with respect to $\hat{\alpha}=(\alpha,\beta) \in U$ (also called Mattson-Solomon polynomial in \cite{Evans}) is the polynomial  $\varphi_{\hat{\alpha},f}=\varphi_{\hat{\alpha},f}(X,Y)=\sum_{(i,j)\in I} f(\alpha^i,\beta^j) X^iY^j\in \Le(r_1,r_2).$ 

It is known that the discrete Fourier transform may be viewed as an isomorphism of algebras  $\varphi_{\hat{\alpha}}:\Le(r_1,r_2)\longrightarrow (\Le^{|I|},\star)$, where the multiplication ``$\star$'' in $\Le^{|I|}$ is defined coordinatewise. Thus, we may see $\varphi_{\hat{\alpha},f}$ as a vector in $\Le^{|I|}$ or as a polynomial in $\Le(r_1,r_2)$ (see \cite[Section 2.2]{Camion}).

As it is usual, we denote by $M=\left(a_{ij}\right)_I$ the matrix indexed by $I$ (or the $I$-matrix) with entries in a ring $R$, and write $a_{ij}=M(i,j)$; in the case of vectors we write $v=(a_i)_{\Z_r}$. For an easy identification of the reader to the results in \cite{BBCS2}, we denote the $i$-th row of $M$ as $H_M(1,i)$  and the $j$-th column of $M$ as $H_M(2,j)$. 

In \cite{BBCS2}, the following definition is used to compute the apparent distance of an abelian code. Let $D\subseteq I$. The matrix afforded by $D$ is defined as $M=\left(a_{ij}\right)_I$ where $a_{ij}=1$ if $(i,j)\not\in D$ and $a_{ij}=0$ otherwise. When $D$ is an union of $q$-orbits we say that $M$ is a $q$-orbit matrix, and it will be denoted by $M=M(D)$. For any $I$-matrix $M$ with entries in a ring, we define the support of $M$ as the set $supp(M)=\left\{(i,j)\in I \tq a_{ij}\neq 0\right\}$, whose complement with respect to $I$ will be denoted by $\D(M)$. Note that, if $D$ is a union of $q$-orbits then the $q$-orbit matrix afforded by $D$ verifies that  $\D(M(D))=D$. Finally, we denote the matrix of coefficients of a polynomial $f\in \F(r_1,r_2)$ by $M(f)$.

Let $\Q$ be the set of all the $q$-orbits in $I$. We define a partial ordering over the set of $q$-orbits matrices $\left\{M(D)\tq D= \cup Q,\text{ for some }Q\in \Q\right\}$ as follows:
\begin{equation}\label{matrixordering}
M(D)\leq M(D') \Leftrightarrow supp\left(M(D)\right)\subseteq supp\left(M(D')\right).
\end{equation}
Clearly, this condition is equivalent to $D'\subseteq D$.

\section{The strong apparent distance and multivariate BCH codes}\label{seccion distancia aparente}

In \cite{BBCS2}, we introduced the notion of strong apparent distance of polynomials and hypermatrices and we applied it to define and study a notion of multivariate BCH bound and BCH abelian codes. As it was pointed out in the mentioned paper, the notion of strong apparent distance was based in the ideas and results in \cite{Camion} and \cite{Evans}. In this section, we recall some notions and results in \cite{BBCS2} restricted to matrices; that is, the bivariate case, because these are the only results that we will use, and it is much simpler to expose. 

For a positive integer $r$, we say that a list of canonical representatives  $b_0,\dots,b_l$ in $\Z_r$ is a list of consecutive integers modulo $r$, if for each $0\leq k<l$ we have that $b_{k+1}= \overline{b_{k}+1}$ in $\Z_r$ and so $b_k=\overline{b_{k+1}-1}$. If $b=b_k$ (resp. $b=b_{k+1}$) we denote $b^+=b_{k+1}$ (resp. $b^-=b_k$).

\begin{definition}\label{hipercolumnas cero consecutivas}
Let $M$ be a matrix over $\F_q$. For any  $k\in \{1,2\}$ and $b\in \Z_{r_k}$, the set of zero rows (if $k=1$) or columns (if $k=2$) of $M$ associated to the pair $(k,b)$ is the set
 \[CH_M(k,b)=\{H_M(k,b_0), \dots , H_M(k,b_l)\}\quad\text{with}\;b=b_0,\]
such that $H_M(k,b_j)=0$ for all $j\in \{0,\dots, l\}$,  $b_0,\dots,b_l$ is a list of consecutive integers modulo $r_k$ and $H_M(k,b_l^+)\neq 0$. We denote by $\omega_M(k,b)$ the value $|CH_M(k,b)|$; in the case of vectors we write $\omega_M(b)=\omega_M(1,b)$.

We define  $\omega_M(k,b)=0$ if $H_M(k,b)\neq 0$.
\end{definition}

\begin{definition}\label{DAF hipermatrices}
Let $q$, $r_1,r_2$ and $I$ be as above and let $M$ be a matrix over $\F_q$. The strong apparent distance of $M$, denoted by $sd^\ast(M)$, is defined as follows
\begin{enumerate}
 \item $sd^\ast(0)=0$.
\item The strong apparent distance of a vector $M=v$ is
\[sd^\ast(v)=\max\{\omega_v(b) +1 \tq b\in \Z_{r}\}.\]
 
\item  In the case of matrices, we proceed as follows, for $k\in \{1,2\}$:
 \begin{eqnarray*}
\epsilon_M(k)&=&\max\{sd^\ast(H_M(k,b)) \tq b\in \Z_{r_k}\} ;\\
\omega_M(k)&=&\max\{\omega_M(k,b) \tq b\in \Z_{r_k}\}.
\end{eqnarray*}
Then

\begin{itemize}
\item[3.1)]  The strong apparent distance of $M$ with respect to the $k$-th variable is 
$sd_k^* (M) =\epsilon_M(k)\cdot(\omega_M(k)+1)$ and
\item[3.2)] the strong apparent distance of $M$ is $sd^\ast (M) =\max_{1 \leq k \leq 2}\{ sd_k^*(M) \}.$
\end{itemize}
\end{enumerate}
 \end{definition}

Now we recall the definition of strong apparent distance of an abelian code in \cite{BBCS2} in the bivariate case.

\begin{definition}\label{apparentdistance}
 Let $C$ be a code in $\F_q(r_1,r_2)$. The strong apparent distance of $C$, with respect to $\hat\alpha \in U$, is $sd_{\hat \alpha}^* (C)= \min\left\{sd^\ast \left(M(\varphi_{\hat\alpha,e})\right) \tq 0\neq e^2 = e \in C\right\}$. The strong apparent distance of $C$ is $sd^\ast (C)= \max\left\{sd_{\hat\beta}^\ast (C) \tq \hat\beta \in U \right\}$. 
 
 We also define the set of optimized roots of $C$ as $\mathcal R(C)=\{\hat\beta \in U \tq sd^\ast (C)=sd_{\hat\beta}^\ast (C)\}.$
\end{definition}

In \cite{BBCS2} it is proved that, for any $f \in C$ and $\hat\alpha \in U$, $sd^*_{\hat\alpha}(C)=\min\{sd^*(M(\varphi_{\hat\alpha,c})) \tq c \in C\}$ and the weight of $f$ verifies $\omega (f) \geq sd^*(M(\varphi_{\hat\alpha,f}))$ (see also \cite{Camion,Evans}); so that, the strong apparent distance of an abelian code is a lower bound for the minimum distance; in fact, the strong apparent distance of any cyclic code, in the obvious sense, is exactly the maximum of all its BCH bounds (what P. Camion calls the BCH bound of an abelian code) \cite[pp. 21-22]{Camion}. Indeed, the following result is proved.

\begin{theorem}{\cite[Theorem 16]{BBCS2}}
 For any abelian code $C$ in $\F_q(r_1,r_2)$ the inequality $sd^\ast(C)\leq d(C)$ holds. 
\end{theorem}

In \cite{BBCS2}, $q$-orbit matrices and coefficient matrices of the images of the discrete Fourier transform of idempotent elements in $\F(r_1,r_2)$ are related, for a fixed $\hat\alpha=(\alpha,\beta)$, as follows. For any idempotent $e\in \F_q(r_1,r_2)$, let $E$ by its generated ideal. Then $M(\varphi_{\hat\alpha,e})=M(\D_{\hat\alpha}(E))$. 
Conversely, any $q$-orbit matrix corresponds with an idempotent; that is, if $P\leq M(\D_{\hat\alpha}(E))$ [see (\ref{matrixordering})] then there exists an idempotent $e'\in E$ such that $P=M(\varphi_{\hat\alpha,e'})$. So it is concluded that the apparent distance of an abelian code $C$ with $M=M(\D_{\hat\alpha}(C))$ may be computed by means of $q$-orbits matrices $P\leq M(\varphi_{\hat\alpha,e})$; that is $\min\{sd^\ast(P)\tq P\leq M\}=\min\{sd^\ast(M(\varphi_{\hat\alpha,e}))\tq e^2=e\in C\}$. This fact drives us to the following definition.

\begin{definition}
 In the setting described above, for a $q$-orbits matrix $M$, the minimum strong apparent distance is
	\[msd(M)=\min\{sd^\ast(P)\tq P\leq M\}.\]
\end{definition}

Finally, one has \cite[Theorem 18]{BBCS2} that for any abelian code $C$ in $\F_q(r_1,r_2)$ with generating idempotent $e$ it happens  $sd_{\alpha}^\ast (C)=msd\left(M(\varphi_{\alpha,e})\right)$ ($\alpha\in U$). Therefore, 
\begin{equation}
 sd^\ast (C)=\max\{msd\left(M(\varphi_{\hat\alpha,e})\right): \hat\alpha\in U\}.
\end{equation}

In \cite[Section IV]{BBCS2} it is presented an algorithm to find, for any abelian code, a list of matrices (or hypermatrices in case of more than 2 variables) representing some of its idempotents  whose strong apparent distances go decreasing until the minimum value is reached. It is a kind of ``suitable idempotents chase through hypermatrices'' \cite[p. 2]{BBCS2}. This algorithm is based on certain manipulations of the $q$-orbit matrix afforded by the defining set of the abelian code. We recall a result of this section that we use repeatedly.

\begin{proposition}{\cite[Proposition 23]{BBCS2}}\label{matrizdamvarias}
Let $D$ be a union of $q$-orbits and $M=M(D) \neq 0$. For $1\leq k\leq 2$ and $b\in \Z_{r_k}$, let $H_M(k,b)$ be a row or column such that $sd^\ast (M)=(\omega_M(k)+1)sd^\ast (H_M(k,b))$; that is, involved row or column in the computation of the strong apparent distance. If $sd^\ast(H_M(k,b))=1$ then $msd(M)=sd^\ast (M)$.
\end{proposition}

Now we recall the definition of multivariate BCH code in the two-dimensional case.

\begin{definition}{\cite[Definition 33]{BBCS2}}\label{codigo bch multivariable}
Let $q$, $r_1,r_2$ and $I$ be as above. Let $\gamma \subseteq \{1,2\}$ and $\delta=\{r_k\geq \delta_k\geq 2\tq k\in\gamma\}$. An abelian code $C$ in $\F_q(r_1,r_2)$ is a bivariate BCH code of designed distance $\delta$ if there exists a list of positive integers $b=\{b_k\tq k\in\gamma\}$ such that
\[\D_{\hat\alpha}(C)=\bigcup_{k\in\gamma} \bigcup_{l=0}^{\delta_k-2}\bigcup_{\b{i}\in I(k,\overline{b_k+l})}Q(\b{i})\]
for some $\hat\alpha \in U$, where $\{\overline{b_k},\dots,\overline{b_k+\delta_k-2}\}$ is a list of consecutive integers modulo $r_k$ and $I\left(k,\overline{b_{k}+l}\right)=\left\{\b{i}\in I\tq \b{i}(k)= \overline{b_{k}+l} \right\}$.

\smallskip

We denote $C=B_q(\alpha,\gamma,\delta,b)$, as usual.
\end{definition}

As a direct consequence of \cite[Theorem 30]{BBCS2} we have $sd^\ast \left(B_q(\hat\alpha,\gamma,\delta,b)\right) \geq \prod_{k\in\gamma}\delta_k$; also, from \cite[Theorem 36]{BBCS2}, we have that 
\[\dim_{\F_q}B_q(\hat\alpha,\gamma,\delta,b)\geq r_1r_2-\lcm\left\{\mathcal{O}_{r_1}(q),\mathcal{O}_{r_2}(q)\right\}\left(\sum_{k\in\gamma} \left(\delta_k-1 \right) \prod_{\doble{j\in\{1,2\}}{j\neq k}} r_j   \right),\] where $\mathcal{O}_{r_k}(q)$ denotes the multiplicative order of $q$ modulo $r_k$. In fact, $B_q(\alpha,\gamma,\delta,b)$ is the abelian code with highest dimension over $\F_q$, whose defining set contains $\cup_{k=1}^2\{\overline{b_k},\dots,\overline{b_k+\delta_k-2}\}$, of the sets described above, as it happens with the (cyclic) BCH codes \cite[Corollary 35]{BBCS2}.

\begin{example}\rm{
Let $C_1$ and $C_2$ be abelian codes in $\F_{2^2}(7,9)$ with defining sets $\D(C_1)=Q(0,0)\cup Q(1,0)\cup Q(3,0)\cup\left(\cup_{t=0}^6 Q(t,1)\right)$ and $\D(C_2)=\D(C_1)\cup Q(0,2)\cup Q(0,3)\cup Q(0,6)$. Then one may check that $C_1=B_{2^2}\left(\hat\alpha,\{2\},\{3\},\{0\}\right)$ and $C_2=B_{2^2}\left(\hat\alpha,\{1,2\},\{2,3\},\{0,0\}\right)$.
}\end{example}

\section{Multiplying dimension in abelian codes}

We shall construct abelian codes starting from BCH (univariate) codes with designed distance $\delta\in \N$. We keep all notation from the preceding sections.

\begin{lemma}\label{lema dimension multiplicada}
 Let $D$ be a union of $q$-orbits modulo $(r_1,r_2)$ and consider the $q$-orbits matrix $M=M(D)$. The following conditions on $M$ are equivalent:
\begin{enumerate}
 \item \label{condicion 1 lema dimension multi} Each column $H_M(2,j)$ verifies that either $H_M(2,j)=0$ or all of its entries have constant value $1$.
\item \label{condicion 2 lema dimension multi} For all $(i,j)\in\Z_{r_1}\times\Z_{r_2}$, it happens that $(i,j) \in D$ if and only if $(x,j) \in D$ for all $x \in \Z_{r_1}$. 
\end{enumerate}
\end{lemma}

\begin{proof}
The result comes immediately from the definition of (hyper)matrix afforded by $D$; that is, for any $a_{ij}\in M$, $a_{ij}=0$ if and only if $(i,j)\in D$ and, otherwise, $a_{ij}=1$.
\end{proof}

As the reader may see, an analogous result may be obtained by replacing $r_2$ by $r_1$. For our next theorem we recall the definition of the set  of optimized roots of $C$ as $\mathcal R(C)=\{\beta \in U \tq sd^\ast (C)=sd_{\beta}^\ast (C)\}$.

\begin{theorem}\label{teorema dimension multiplicada}
 Let $n$ and $r$ be positive integers such that $\gcd(q,nr)=1$.  Let $C$ be a nonzero cyclic code in $\F_q(r)$ with $sd^\ast (C)= \delta >1$ and $\hat\alpha=(\alpha_1,\alpha_2) \in U_n \times \mathcal{R}(C)$. Then, the abelian code $C_n$ in $\F_q(n,r)$ with defining set $\D_{\hat\alpha}(C_n)=\Z_n\times \D_{\alpha_2}(C)$ verifies that $sd^\ast (C_n) = \delta$ and $\dim_{\F_q}(C_n)=n\dim_{\F_q}(C)$.
\end{theorem}

\begin{proof}
Consider any $\hat\beta = (\beta_1, \beta_2) \in U_n\times U_r$ and let $C_n$ be the abelian code such that $\D_{\hat\beta}(C_n) = \Z_n\times \D_{\beta_2}(C)$. It is clear that $\D_{\hat\beta}(C_n)$ satisfies the condition~\textit{(\ref{condicion 2 lema dimension multi})} of Lemma \ref{lema dimension multiplicada}; so, the $q$-orbits matrix afforded by $\D_{\hat\beta}(C_n)$, say $M = M(\D_{\hat\beta}(C_n))$, verifies the condition~\textit{(\ref{condicion 1 lema dimension multi})} of that lemma. Clearly, since $C\neq 0$ there is at least one nonzero column. If  $N = (a_j)_{j\in \Z_r}$ is the $q$-orbit vector afforded by $\D_{\beta_2}(C)$ then $H_M(2, j) = 0$ if and only if $a_j = 0$. (We only focus in $k=2$ on view of Proposition~\ref{matrizdamvarias}.)

So, $\omega_M(2,b)=\omega_N(1,b)$ for all $b\in \Z_r$ and hence $\omega_M(2)=sd^\ast(N)\leq \delta$, and the equality is reached when $\beta_2\in \mathcal R(C)$. On the other hand, we have that $sd^\ast(H_M(2,b))=1$, for any nonzero column; so that, $\epsilon_M(2)=1$, and this happens for any element of $ U_n$. Hence $sd^\ast(M)\leq \delta$.  Now, by Proposition~\ref{matrizdamvarias},  $sd^\ast_{\hat\beta} (C_n) = msd(M) = sd^\ast (M)\leq \delta$, and the equality is reached if $\beta_2\in \mathcal R(C)$; so that $sd^\ast(C_n) =\delta$.

 Finally, since $\dim_{\F_q} (C_n) = |supp(M)|$, we have that $\dim_{\F_q} (C_n) = n \dim_{\F_q} (C)$.
\end{proof}

Now we shall see that this multiplying dimensions technique extends BCH (cyclic) codes to multivariate BCH abelian codes.

\begin{corollary}
 In the setting of theorem above, if $C$ is a BCH code $C=B_q(\alpha_2,\delta,b)$ then $C_n=B_q\left((\alpha_1,\alpha_2),\{2\},\{\delta\},\{b\}\right)$
\end{corollary}
\begin{proof}
 Immediate from Definition~\ref{codigo bch multivariable} and the theorem above.
\end{proof}

\begin{example}\rm{
 Set $q=2$, $r=55$, $n=3$, $\hat\alpha=(\alpha_1,\alpha_2) \in U_{3}\times U_{55}$ and let  $C$ be the cyclic code in $\F_2(55)$ with defining set with respect to $\alpha_2$, $D=\D_{\alpha_2}(C)=C_2(1)\cup C_2(5)$. Set $M=M(D)$. A simple inspection on $M$ shows us that $C$ is a BCH code with parameters $C=B_2(\alpha_2,7,13)$ and dimension 25. By the corollary above, we may construct the new bivariate code $C_3$ with defining set $\D_{\hat\alpha}(C_3)=\Z_3 \times D$. So that $C_3= B_2(\hat\alpha,\{2\},\{7\},\{13\})$ in $\F_2(3,55)$, $sd^\ast (C_3)=7$ and $\dim_{\F_2}(C_3)=75$.
}\end{example}

In \cite{BGS} we determine some types of abelian codes whose minimum distance equals their strong apparent distance. We now apply these techniques to go further and construct multiplied dimensional abelian codes with the same property.

\begin{proposition}
  Let $n$ and $r$ be positive integers with $\gcd(q,nr)=1$ and let $C$ be a nonzero cyclic code in $\F_q(r)$ such that  $sd^\ast(C)=d(C)$. Then there exists $\hat\alpha=(\alpha_1,\alpha_2) \in U_n \times \mathcal{R}(C)$ such that the abelian code $C_n$ in $\F_q(n,r)$ with defining set $\D_{\hat\alpha}(C_n)=\Z_n\times \D_{\alpha_2}(C)$ verifies the equality  $d(C_n)=d(C)$.
\end{proposition}

\begin{proof}
 Since $sd^\ast(C)=d(C)$ then, by \cite[Theorem 1]{BBCS true} there exists $b'\in \Le(r)$ and $h\in\Z_r$ such that $Y^{h}b'\mid Y^r-1$ and there exists $\alpha_2\in U_r$ such that $\varphi_{\alpha_2,b'}^{-1}\in C$. Set $b=Y^{h}b'$. On the other hand, setting $a=\sum_{i=0}^{n-1}X^i$, we have that, $a\mid X^n-1$ and that $\varphi_{\alpha_1,a}^{-1}\in \F_q(n)$, for any $\alpha_1\in U_n$. Now, a direct aplication of \cite[Theorem 36]{BGS} gives us that $sd^\ast (C_n)=d(C_n)$ and hence by the theorem above $d(C_n)=d(C)$.
\end{proof}

Combinning this proposition with the construction techniques in \cite{BBCS true} we may find a number of examples of Abelian Codes $C$, satisfying the equality $sd^\ast(C)=d(C)$. For example, from \cite[Corollary 7]{BBCS true}, we know that in the case $r=2^m-1$, for any $m\in\N$, for each $n\in \N$ there are at least $\frac{\phi(r)}{m}$ binary multiplied-dimension abelian codes $C_n$ such that $sd^\ast (C_n)=d(C_n)$.

Let us multiply dimension of some famous BCH codes. It is known (see \cite{S}) for any $h,m\in\N$, if $C$ is a BCH code of length $r=q^m-1$ and designed distance $\delta=q^h-1$ over $\F_q$ then $d(C)=\Delta(C)$. As a direct consequence we have

\begin{corollary}
 Let $h,m\in\N$. If $C$ is a BCH code of length $r=q^m-1$ and designed distance $\delta=q^h-1$ over $\F_q$ then $d(C_n)=sd^\ast(C_n)$, for any $n\in\N$, with $gcd(n,q)=1$.
\end{corollary}

In order to multiply dimension from multivariate Reed Solomon codes we have the following result that we present in the multivariate case instead of bivariate case.

\begin{proposition}\label{multiplicar dimension una cara}
Let  $B_q(\hat\alpha,\gamma,\delta,b)$ be a multivariate BCH code with $\gamma=\{k\}$, $\delta=\{\delta_k\}$ and $b=\{b_k\}$, for some $k\in \{1,\dots,s\}$. If $r_k=q-1$ then we have $sd_{\hat\alpha}^\ast\left(B_q(\hat\alpha,\gamma,\delta,b)\right)=\delta_k$ and $\dim_{\F_q}(B_q(\hat\alpha,\gamma,\delta,b))=\left(r_k-\delta_k+1\right)\prod_{\doble{j=1}{j\neq k}}^s r_j$.
\end{proposition}

\begin{proof}
Set $C=B_q(\hat\alpha,\gamma,\delta,b)$ and suppose that $C\neq 0$.  Since $r_k=q-1$ we have that $l q \equiv l \mod r_k$, for all $l\in \Z_{r_k}$; hence $Q(\b{i}) \subseteq I(k,l)$  for all $\b{i} \in I(k,l)$.  So, $$\D_{\hat\alpha}(C)=\bigcup_{l=\overline{b_k}}^{\overline{b_k+\delta_k-2}} I(k,l).$$ 

Since $\dim_{\F_q}(C)=\prod_{j=1}^s r_j-|\D_{\hat\alpha}(C)|$ we have that $\dim_{\F_q}(C)= \left(r_k-\delta_k+1\right)\prod_{\doble{j=1}{j\neq k}}^s r_j$. Note that if $M=M(\D_{\hat\alpha}(C))$ and $l \in \{\overline{b_k}, \dots, \overline{b_k+\delta_k-2}\}$ then $H_M(k,l)=0$ and all nonzero hypercolumns have entries with constant value 1. Therefore, $\omega_M(k,\overline{b_k})=\delta_k-1$ and $\epsilon_M(k)=1$, which give us $sd_k^\ast (M)=\delta_k$. 

Moreover, following the terminology of Proposition~\ref{matrizdamvarias}, every nonzero hypercolumn $H_M(k,t)$ is an involved hypercolumn and have entries with constant value 1, so that by such proposition $msd(M)=sd^\ast (M)=\delta_k$. Therefore,  $sd_{\hat\alpha}^\ast (B_q(\hat\alpha,\gamma,\delta,b))=\delta_k$.
\end{proof}

The proposition above is applicable to codes that we obtain by using the construction given in Theorem~\ref{teorema dimension multiplicada}, when we start from Reed-Solomon codes. The following proposition allows us to compute dimension and true minimum distance.

\begin{corollary}
Let $R=B_q(\alpha,\delta,b)$ be a Reed-Solomon code of length $r$. Then, for each positive integer $n$ and any $\alpha'\in U_n$, there exists a multivariate BCH code, $C=B_q\left((\alpha',\alpha),\{2\},\{\delta\},\{b\}\right)$, such that $\dim(C)=\left(r-\delta+1\right)n= n \cdot \dim_{\F_q}(R)$ and $d(C)=sd_{\hat\alpha}^\ast(C)=\delta$.
\end{corollary}
\begin{proof}
 Apply \cite[Section 5.2]{BGS} to results above.
\end{proof}

As a direct consequence of the corollary above, we may conclude that abelian codes that are multiplied Reed-Solomon codes are never MDS codes. In fact, we do not know MDS abelian codes having minimum distance greater than 1.

\begin{example}
 A popular Reed Solomon code is the $RS(255,223)$ (see \cite{LyM}), in our notation $R=B_{2^8}(\alpha,33,0)$. By multiplying its dimension by $5$, we get  $C_5=B_{2^8}\left((\alpha',\alpha),\{2\},\{33\},\{0\}\right)$, such that $\dim(C_5)=1115$ and $d(C_5)=33$.
\end{example}

An interesting task would be to find decoding methods for these kinds of codes.

\section{Conclusion}

In this note, we gave a particular construction of bivariate Abelian Codes from cyclic codes,  multiplying their dimension and preserving their apparent distance. In the case of cyclic codes whose maximum BCH bound equals its minimum distance the obtained abelian code verifies the same property; that is, the strong apparent distance and the minimum distance coincide.  This construction may be used to multiply Reed-Solomon codes to abelian codes.

% \subsection*{Acknowledgment}
% Many thanks to our \TeX-pert for developing this class file.
%

%  --------------------------------------------------
\end{document}